\documentclass[runningheads,english,envcountsame]{llncs} %numberwithinsect

\usepackage{graphicx}
\usepackage[utf8]{inputenc}
\usepackage[T1]{fontenc}
\usepackage{subfig}
\usepackage{color}
\usepackage[dvipsnames,table]{xcolor}
\usepackage{amssymb,amsfonts, amsmath,thm-restate}
\usepackage{paralist}
\usepackage{xspace}
\usepackage{wrapfig}
\usepackage[ruled, vlined, linesnumbered]{algorithm2e}
\newcommand{\Arr}{\ensuremath{A}\xspace} % 
\newcommand{\CurArr}{\ensuremath{\mathcal{A}\xspace}} % 

\newcommand{\CurX}{\ensuremath{\mathcal{X}\xspace}} % 
\newcommand{\CurL}{\ensuremath{\mathcal{L}\xspace}} % 
\newcommand{\CurO}{\ensuremath{\mathcal{O}\xspace}} % 
\newcommand{\CurQ}{\ensuremath{\mathcal{Q}\xspace}} % 
\newcommand{\CurH}{\ensuremath{\mathcal{H}\xspace}} % 
\newcommand{\Cell}{\ensuremath{\mathcal{C}\xspace}} % 

\newcommand{\LX}{\ensuremath{X}\xspace} % 
\newcommand{\LO}{\ensuremath{O}\xspace} % 
\newcommand{\LQ}{\ensuremath{Q}\xspace} % 

\newcommand{\N}{\ensuremath{\mathbb{N}}\xspace}

\newcommand{\cupdot}{\mathbin{\mathaccent\cdot\cup}}

\DeclareMathOperator{\sentinel}{\bot} % sentinel for edge types that do not exits

\newcommand{\myparagraph}[1]{\noindent\textbf{#1.}\,}

\title{Towards a characterization of stretchable aligned graphs\thanks{Work was partially supported by grant RU
1903/3-1 of the German Research Foundation(DFG).}}

\author{
Marcel Radermacher\inst{1}
\and Ignaz Rutter\inst{2} 
\and Peter Stumpf\inst{2}
}

\institute{Department of Informatics, Karlsruhe Institute of
Technology (KIT), Germany
\and Faculty of Computer Science and Mathematics, University of
Passau, Germany
\email{radermacher@kit.edu},
\email{\{rutter,stumpf\}@fim.uni-passau.de}
}

\begin{document}
\maketitle

\begin{abstract}
	We consider the problem of stretching pseudolines in a planar straight-line
	drawing to straight lines while preserving the straightness and the
	combinatorial embedding of the drawing.  
	We answer open questions by Mchedlidze  et
	al.~\cite{DBLP:journals/jgaa/MchedlidzeRR18} by showing that not all instances
	with two pseudolines are stretchable. On the positive side, for $k\geq 2$
	pseudolines intersecting in a single point, we prove that in case that some
	edge-pseudoline intersection-patterns are forbidden, all instances are
	stretchable.
	For intersection-free pseudoline arrangements we show that every aligned graph
	has an aligned drawing.
	This considerably reduces the gap between stretchable and non-stretchable
	instances.
\end{abstract}

%\setcounter{page}{1}

%% ==============================
\section{Introduction}
\label{ch:introduction}
%% ==============================

\begin{figure}[t]
	\centering
	\subfloat[\label{fig:oneline}]{ \includegraphics{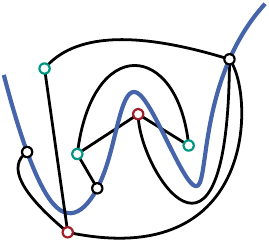} }
	\subfloat[\label{fig:pappus}]{ \includegraphics{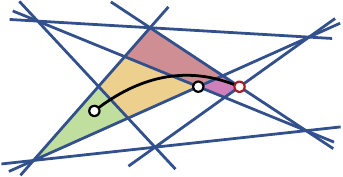} }
	\quad
	\subfloat[\label{fig:alignment_complexity:1}]{\includegraphics[page=1]{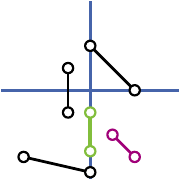}}
	\quad
	\subfloat[\label{fig:alignment_complexity:2}]{ \includegraphics[page=2]{figures/alignment_complexity.pdf} }

	\caption{%Throughout the paper, a blue curve indicate pseudolines.
		(a) An aligned graph on one (blue) pseudoline. The color indicates the
		vertex partition $L\cup R\cup S$.
		(b) Aligned graph of alignment
		complexity $(\bot,3,\bot)$ that does not have an aligned
		drawing~\cite{DBLP:journals/jgaa/MchedlidzeRR18}.
		(c) Allowed
		types of edges in aligned graphs of alignment complexity $(1,0,0)$. The
		green edge is aligned. The purple edge is free. (d) Aligned graph of
		alignment complexity $(2,1, \bot)$.   }
\end{figure}

%By the classical result of Tutte~\cite{tutte1963draw} 
Every planar graph $G=(V,E)$ has a straight-line
drawing~\cite{MR26311,tutte1963draw}.  In a restricted setting one seeks a
drawing of $G$ that obeys given constraints,
%Hence, it is natural to ask whether $G$ has a straight-line drawing  that obeys
%prescribed constraints. 
%
e.g., Biedl et al.~\cite{Biedl1998,Biedl:1998:DPP:276884.276917} studied whether
a bipartite planar graph has a drawing where the two sets of the partitions can
be separated by a straight line. Da Lozzo et
al.~\cite{DBLP:journals/jocg/LozzoDFMR18} generalized this result and
characterized the planar graphs with a partition $L \cup R \cup S = V$ of the
vertex set that have a planar straight-line drawing such that the vertices in
$L$ and $R$ lie left and right of a common line $l$, respectively, and the
vertices in $S$ lie on $l$;  refer to Fig.~\ref{fig:oneline}.  In this case $S$ is called \emph{collinear}. In
particular, they showed that $S$ is collinear if and only if there is a drawing
of $G$ such that there is an open simple curve $\mathcal P$ that starts and ends
in the outer face of $G$, separates $L$ from $R$, collects all vertices in $S$
and that either entirely contains or intersects at most once each edge.
%where each edge either entirely lies on $\mathcal P$ or has at most one
%intersection with $\mathcal P$.  
%
We refer to $\mathcal P$ as a \emph{pseudoline with respect to~$G$}.
%or a \emph{proper %$S$-good curve}. 

Dujmovic et al.~\cite{DBLP:conf/soda/DujmovicFGMR19}
proved the following surprising result: If $S$ is a collinear set, then for
every point set $P$ with $|S| = |P|$ there is a straight-line drawing $\Gamma$
of $G$ such that $S$ is mapped to $P$. Another recent research stream considers
the problem of drawing all vertices on as few lines as
possible~\cite{DBLP:conf/gd/ChaplickFLRVW16}.
Eppstein~\cite{DBLP:journals/corr/abs-1903-05256} proved that for every integer
$l$ there is a cubic planar graph graph $G$ with $O(l^3)$ vertices such that not
all vertices of $G$ can lie on $l$ lines.

Mchedlidze et al.~\cite{DBLP:journals/jgaa/MchedlidzeRR18} generalized the
concept of a single pseudoline with respect to an embedded graph 
to an arrangements of pseudolines and introduced the notion of
\emph{aligned graphs}, i.e, a pair $(G, \mathcal A)$ where $G$ is a planar
embedded graph and $\CurArr = \{\CurL_1, \dots, \CurL_k\}$ is a set of
pseudolines $\CurL_i$ with respect to $G$
that intersect pairwise at most once. 
We cite the original definition of aligned
drawings~\cite{DBLP:journals/jgaa/MchedlidzeRR18}.  A tuple $(\Gamma,A)$ of a
(straight-line) drawing $\Gamma$ of $G$ and line arrangement $A$ is an
\emph{aligned drawing of $(G, \CurArr)$} if and only if the arrangement of the
union of $\Gamma$ and $A$ has same combinatorial properties as the union of $G$
and $\CurArr$. In the following, we specify these combinatorial properties. Let
$A = \{L_1, L_2, \dots, L_k\}$, i.e., line $L_i$ corresponds to pseudoline
$\CurL_i$.
A (pseudo)-line arrangement divides the plane into a set of \emph{cells}
$\Cell_1, \Cell_2, \dots, \Cell_\ell$. If $A$ is homeomorphic to $\CurArr$, then
there is a bijection $\phi$ between the cells of $\CurArr$ and the cells of
$\Arr$.
If  $(\Gamma, A)$ is an aligned drawing of $(G, \CurArr$), then it has the
following properties:
\begin{inparaenum}[(i)] \item the arrangement of $A$ is homeomorphic to the
		arrangement of $\CurArr$,
	\item $\Gamma$ is a straight-line drawing homeomorphic to the planar embedding of $G$,
	\item the intersection of each vertex~$v$ and each edge $e$ with a cell
		$\Cell$ of $\CurArr$ is non-empty if and only if the intersection of $v$ and
		$e$ with $\phi(\Cell)$ in $(\Gamma, \Arr)$, respectively, is non-empty,
	\item if an edge $uv$ (directed from $u$ to $v$) intersects a sequence of
		cells $\Cell_1, \Cell_2, \dots, \Cell_r$ in this order, then $uv$ intersects
		in $(\Gamma, \Arr)$ the cells $\phi(\Cell_1), \phi(\Cell_2), \dots,
		\phi(\Cell_r)$ in this order, and 
	\item each line $L_i$ intersects in $\Gamma$ the same vertices and edges as
		$\CurL_i$ in $G$, and it does so in the same order.
%\end{compactenum}
\end{inparaenum}

Mchedlidze et al. observed that not every aligned graph has an aligned drawing.
For example, the modification of the Pappus configuration in
Fig.~\ref{fig:pappus} does not have an aligned drawing. Note that one endpoint
of the edge is \emph{anchored} on some pseudolines and that the edge
\emph{crosses} three pseudolines.  Hence, Mchedldize et al.  studied a
restricted subclass of aligned graphs that only contains edges $uv$ that are
either (see Fig.~\ref{fig:alignment_complexity:1} and
Fig.~\ref{fig:alignment_complexity:2})
\begin{itemize}
	\item \emph{free}, i.e, the entire edge $uv$ is in a single cell,
	\item \emph{aligned}, i.e., the entire edge $uv$ is on a single pseudoline,
	\item \emph{one-sided anchored}, i.e., $u$ or $v$ is on a pseudoline but not
		both, and $uv$ does not cross a pseudoline,
	\item \emph{1-crossed}, i.e., $u$ and $v$ are in the interior of a cell and
		$uv$ crosses one pseudoline.
\end{itemize}

For this restricted class Mchedlidze et al. proved that every aligned graph has
an aligned drawing. For this purpose they reduced their instances to aligned
graphs that do neither have free edges nor aligned edges nor separating
triangles. Then the original instance has an aligned drawing if the reduced
instance has an aligned drawing. Thus, the key to success is to characterize the
reduced instances and to prove that every reduced instance has an aligned drawing.
In the reduced setting, Mchedlidze et al. were able to show that each cell of
the pseudoline arrangement contains at-most a single vertex. Since the union of
two adjacent cells in the line arrangement is convex, any placement of the
vertices that respects the ordering constraints along the lines induces a valid
aligned drawing of the reduced aligned graph.  If we additionally allow
\emph{two-sided anchored edges}, i.e., edges where both endpoints are on
pseudolines but that do not cross a pseudoline, then it is possible to construct
a family of aligned graphs such that each cell can contain a number of vertices
that is not bounded by the number of pseudolines.

\myparagraph{Contribution} We show that every aligned graph on $k\geq 2$
pseudolines intersecting in a single point with free, aligned, one-sided and
two-sided anchored, and 1-crossed edges has an aligned drawing.  If we allow an
additional edge type, we show that there is an aligned graph on two pseudolines
that does not have an aligned drawing. Note that in the example given in
Fig.~\ref{fig:pappus}, no point in the green cell is visible from the red vertex
within the polygon defined by union of the (colored) cells traversed by the
edge. Hence, this instance trivially does not admit an aligned drawing. In
contrast, each edge in Fig.~\ref{fig:counterexample} can be drawn independently
as a straight-line segment. We show that the entire instance does not admit a
straight-line drawing. 
Further, we show that every aligned graph $(G, \CurArr)$ has an aligned drawing,
if $\CurArr$ does not have crossings, i.e., $\CurArr$ corresponds to an
arrangement $\Arr$ of parallel lines. This couples aligned graphs to
hierarchical (level) graphs.
This significantly narrows the gap in the characterization of realizable and
non-realizable aligned graphs.

%% ==============
\section{Preliminaries}
\label{sec:preliminaries}
%% ==============

\begin{figure}
	\centering
	\subfloat[] {
		\includegraphics[page=2]{figures/k_star.pdf}
	}
	\subfloat[] {
		\includegraphics[page=1]{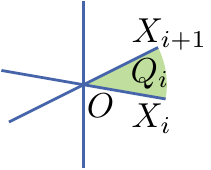}
	}
	\caption{(a,b) (Pesudo)-line arrangements of a $3$-star aligned graph. The
	green region indicates a cell.}
	\label{fig:k_star}	
\end{figure}

We first introduce some notation used for aligned graphs on $k$
pseudolines intersecting in a single point.  Let $\mathcal O$ be a point called
the \emph{origin}.  Let $\CurX = \{\CurX_1, \CurX_2, \dots, \CurX_k\}$ be a
pseudoline arrangement where the pseudolines pairwise intersect in $\mathcal O$;
refer to Fig.~\ref{fig:k_star}. We refer to an aligned graph $(G, \CurX)$ as a
\emph{$k$-star aligned graph}.  Correspondingly, we refer to $(\Gamma, \LX)$,
with $X= \{X_1, X_2, \dots, X_k\}$ as an aligned drawing of $(G, \CurX)$, where
the lines in $X$ pairwise intersect in the \emph{origin} $O$.
The curves in $\CurX$ divide the plane into a set of cells $\CurQ_1, \dots ,
\CurQ_{2k}$ in counterclockwise order.  These cells naturally correspond to the
regions $\LQ_1, \dots, \LQ_{2k}$ bounded by the lines in $\LX$.

We refer to an edge (vertex) as \emph{free} if it is entirely in the interior of
a cell. An \emph{aligned edge (vertex)} is entirely on a pseudoline. For each
\emph{$l$-crossed} edge $e$ there are $l$ but not $l+1$ pseudolines that
intersect $e$ in its interior. An edge $e$ is \emph{$i$-anchored} if $i$ of its
endpoints lie on $i$ distinct pseudolines.
Mchedlidze et al. used a triple $(l_0, l_1, l_2)$, with $l_i \in \N \cup
\{\sentinel\}$ to describe the complexity of an aligned graph $(G, \CurArr)$.
Let $E_i$ be the set of $i$-anchored edges; note that, the set of edges is the
disjoint union $E_0 \cupdot E_1 \cupdot E_2$.  A non-empty edge set $A \subset
E$ is $l$-crossed if $l$ is the smallest number such that every edge in $A$ is
at most $l$-crossed. An aligned graph $(G, \CurArr)$ has alignment complexity
$(l_0, l_1, l_2)$, if $E_i$ is at most $l_i$-crossed or has to be empty, if $l_i
= \sentinel$. In particular, Mchedlidze et al. proved that every aligned graph
of alignment complexity $(1,0,\sentinel)$ has an aligned drawing. Our results
can be restated as that every $2$-star aligned graph of alignment
complexity $(1,0,0)$ has an aligned drawing. Further, there is an aligned graph
of alignment complexity $(\sentinel, 1,\sentinel)$ that does not have an aligned
drawing. 
 
%\begin{lemma}[Mchedlidze et al.~\cite{DBLP:journals/jgaa/MchedlidzeRR18}]
%  For every aligned graph $(G, \CurArr)$ there is an aligned graph $(G',
%  \CurArr)$ that does neither have free edges nor aligned edges nor separating
%  triangles such that $(G, \CurArr)$ has an aligned drawing if $(G', \CurArr)$
%  has an aligned drawing.
%\end{lemma}

%\begin{lemma}[Mchedlidze et al.~\cite{DBLP:journals/jgaa/MchedlidzeRR18}]
%  %
%  For every aligned graph $(G, \CurArr)$ of alignment complexity $(1, 0,
%  \sentinel)$ there is an aligned triangulation $(G', \CurArr)$ of alignment
%  complexity $(1,0, \sentinel)$ such that $G$ is
%  a subgraph of $G'$.
%  %
%\end{lemma}

\section{Star aligned graphs}

\begin{figure}[tb]
  \centering
	\subfloat[\label{fig:counterexample}
]{
		\includegraphics[page=2]{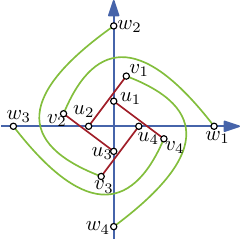}
	}
	\quad
	\subfloat[\label{fig:inequality}]{
  
		\includegraphics[page=3]{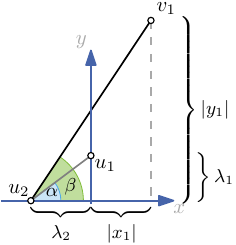}
	}
	\caption{(a)~A $2$-aligned graph that does not have an aligned drawing.
		(b)~We have $\lambda_1/\lambda_2 = \tan(\alpha) < \tan(\beta) =
		{|y_1|}/({\lambda_2+|x_1|})$.}
 \end{figure}

In this section, we study whether $k$-star aligned graphs have aligned drawings. We first
prove that the $2$-star aligned graph in Fig.~\ref{fig:counterexample} does not have
an aligned drawing.

\begin{theorem}\label{theorem:counterexample}
	% 
	%The $2$-aligned graph depicted in Fig.~\ref{fig:counterexample} does not have
	%an aligned drawing.
	%There is a $2$-aligned graph that does not have an aligned drawing.
%=======
	There is a $2$-star aligned graph of alignment complexity $(\bot,1,\bot)$ that does not have
	an aligned drawing.
\end{theorem}
%\iftrue
\begin{proof}
	Assume that the aligned graph in Fig.~\ref{fig:counterexample} has an aligned
	drawing.
	For $i=1,\dots, 4,5$ with $1=5$, let $(x_i,y_i)$ be the point for $v_i$, let
	$\lambda_i$ be the distance of $u_i$ to the origin $\LO$.
	Since $u_2v_1$ intersects the
	$y$-axis above $u_1$, edge $u_2v_1$ has a steeper slope than the segment
	$u_2u_1$; see Fig.~\ref{fig:inequality}.
	%Since $u_2v_1$ intersects the $y$-axis above $u_1$, we obtain
	%\marcel{woher
	%kommt die Gleichung? Bild.}
 %   $\frac{|y(u_1)|}{|x(u_2)|} < \frac{|y_1|}{|x(u_2)|+|x_1|}$
 %and therefore 
 %   $\frac{|y(u_1)|}{|x(u_2)|}\cdot |x_1| < |y_1|$.
	We obtain
	$\lambda_1/\lambda_2 < |y_1|/(\lambda_2+|x_1|)$
 and therefore 
    $|x_1| < \lambda_2 / \lambda_1 \cdot |y_1|$.
		Analogously, we obtain 

  %\begin{equation}
  %  \frac{|y(u_1)|}{|x(u_2)|} \cdot {|x_1|}<{|y_1|},\quad
  %  \frac{|x(u_2)|}{|y(u_3)|} \cdot {|y_2|}<{|x_2|},\quad
  %  \frac{|y(u_3)|}{|x(u_4)|} \cdot {|x_3|}<{|y_3|},\quad
  %  \frac{|x(u_4)|}{|y(u_1)|} \cdot {|y_4|}<{|x_4|}
  %  \label{eq:inner}.
  %\end{equation}

	\begin{equation}
		{|x_i|}<\frac{\lambda_{i+1}}{\lambda_i}\cdot{|y_i|},i=1,3 \qquad
		{|y_i|}<\frac{\lambda_{i+1}}{\lambda_i}\cdot{|x_i|},i=2,4
		%{|x_1|}<\frac{\lambda_{2}}{\lambda_1}\cdot{|y_1|},
    %\frac{|x(u_2)|}{|y(u_3)|} \cdot {|y_2|}<{|x_2|},
    %\frac{|y(u_3)|}{|x(u_4)|} \cdot {|x_3|}<{|y_3|},
    %\frac{|x(u_4)|}{|y(u_1)|} \cdot {|y_4|}<{|x_4|}
    \label{eq:inner}.
  \end{equation}

	%Since $v_2w_1$, $v_3w_2$, $v_4w_3$ and $v_1w_4$ are embedded as straight
	Since $v_{i+1}w_i$ are embedded as straight
	lines, we further obtain estimation (2) that $|y_i| < |y_{i+1}|$ for $i=1,3$ and $|x_i| <
	|x_{i+1}|$ for $i=2,4$.
  %\begin{equation}
  %  |y_1|<|y_2| ,\quad |x_2|<|x_3| ,\quad|y_3|<|y_4| ,\quad|x_4|<|x_1|.
  %  \label{eq:outer}
  %\end{equation}
  By multiplying the left and the right sides we obtain $|x_1|\cdot |y_2|\cdot |x_3|\cdot |y_4|
	\overset{\eqref{eq:inner}}{<} |y_1|\cdot |x_2|\cdot |y_3|\cdot
  |x_4|\cdot\frac{\lambda_2\lambda_3\lambda_4\lambda_1}{\lambda_1\lambda_2\lambda_3\lambda_4}= |y_1|\cdot |x_2|\cdot |y_3|\cdot
	|x_4|\overset{(2)}{<}|y_2|\cdot |x_3|\cdot |y_4|\cdot |x_1|$. A contradiction.
\end{proof}
%\elif

%The crux of Fig.~\ref{fig:counterexample} is that the source of red (green)
%edges is free (aligned) in counterclockwise direction.  In the following 
%we formalize this and show that \emph{counterclockwise aligned
%graphs} (without red edges) have aligned drawings. 

\subsection{Aligned drawings of counterclockwise star aligned graphs}

We now consider aligned drawings of $k$-star aligned graphs $(G, \CurArr$) for $k\geq 2$.
Recall that the aligned graph in Figure~\ref{aligned:fig:counterexample_2}
does not have an aligned drawing. The crux is that the source of the red edges
are free and the source of green edges are aligned.  In the following we
introduce so-called \emph{counterclockwise aligned graphs} and show that they
have aligned drawings. 

\begin{figure}[tb]
	\subfloat[\label{aligned:fig:counterexample_2}]{
		\includegraphics[page=4]{figures/conterexample2.pdf}
	}
	\quad
	\subfloat[]{
		\includegraphics[page=2]{figures/ccw.pdf}
	}
	\quad
	\subfloat[]{
		\includegraphics[page=1]{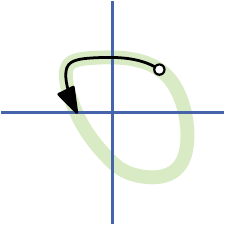}
	}
	\quad
	\subfloat[\label{aligned:fig:def:comb}]{
		\includegraphics{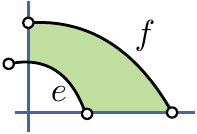}
	}
	\caption{(a)~This $2$-aligned graph does not have an aligned drawing. (b,c)~The green curve indicates the Jordan curve that completes the
		black edge. The edge in (b) is an edge of a ccw-aligned graph. The edge depicted
		in (c) is forbidden in ccw-aligned graphs.
	(d)~A comb of edges $e,f$.}
\end{figure}

We orient each non-aligned edge $uv$ of an aligned graph $(G, \CurX)$ such that
it can be extended to a Jordan curve, i.e., a closed simple curve, $\mathcal
C_{uv}$ with the property that it intersects  each pseudoline exactly twice and
has the origin to its left.  A \emph{counterclockwise aligned (ccw-aligned)}
graph is a $k$-star aligned graph of alignment complexity $(1,1,0)$ whose orientation
does not contain $1$-anchored $1$-crossed edges with a free source vertex. 

We prove that every ccw-aligned graph has an aligned drawing.
To prove this statement we follow the same proof strategy as Mchedlidze et al.
In particular, we have to augment our aligned graph to a particular ccw-aligned
triangulation.  Further, we use that for each
aligned graph $(G, \CurX)$ there is a \emph{reduced aligned graph} $(G_R,
\CurX)$ (i.e., it does neither contain
\begin{inparaenum}[(i)]
	\item separating triangles, nor 
	\item free edges, nor
	\item aligned edges that are not incident to the origin $\mathcal O$)
\end{inparaenum}
with the property that $(G, \CurX)$ has an aligned drawing if $(G_R, \CurX)$ has
an aligned drawing;  see Lemma~\ref{lemma:reduced}. In contrast to aligned graphs of alignment complexity
$(1,0,\sentinel)$ the size of $(G_R, \CurX)$ is not bounded by a constant. The
aim of Lemma~\ref{lemma:base_case:structure} and Lemma~\ref{lem:closedEmpty} is
to describe the structure of the reduced instances. This helps
to prove Lemma~\ref{lemma:base_case:drawing} that states that each reduced
instance has an aligned drawing. 

We first introduce further notations.
A $k$-star aligned graph $(G,
\CurX)$ is a \emph{proper $k$-star aligned triangulation} if each inner face is a
triangle, the boundary of the outer face is a $2k$-cycle of $2$-anchored edges,
the outer face does not contain the origin and there is a degree-$2k$ vertex $o$
on the origin incident to $2k$ aligned edges.  We refer to a reduced proper
ccw-aligned triangulation as a \emph{reduced aligned triangulation}.
We refer to $1$-anchored $1$-crossed and $2$-anchored edges as
\emph{separating}. The region within a cell that is bounded by two
separating edges $e$ and $f$ is an \emph{edge region} (Fig.~\ref{aligned:fig:def:comb}).  An inclusion-minimal edge region is a
\emph{comb}.

The following lemma is a consequence from the results by Mchedlitze et
al.~\cite{DBLP:journals/jgaa/MchedlidzeRR18}. For further details we refer to
the Appendix.

\begin{lemma}\label{lemma:reduced}
	Every $k$-star aligned graph has an aligned drawing, if every reduced $k$-star
	aligned triangulation has an aligned drawing.
\end{lemma}

Hence, our main contribution is to characterize reduced $k$-star aligned
triangulations and then, to prove that every such instance has an
aligned drawing.

\begin{lemma} \label{lemma:base_case:structure}
	Let $(G_R, \CurX)$ be a reduced aligned triangulation and let $o$ be the
	vertex on the origin. Then in $(G_R - o, \CurX)$ each pseudoline $\CurX_i$
	alternately intersects vertices and edges, and each comb contains at most
	one vertex.
\end{lemma}

\begin{proof}
	Assume that there are two consecutive aligned vertices $u$ and $v$. Since $G_R$ is
	triangulated and $u$ and $v$ are consecutive, $G_R$ contains the edge $uv$. This
	contradicts the assumption that $(G_R, \CurX)$ does not contain aligned edges.

	The following modification helps us to prove that there are no two consecutive
	edges along a pseudoline and that no comb contains two free vertices.

	\begin{figure}
		\centering
		\subfloat[] {
			\includegraphics[page=1]{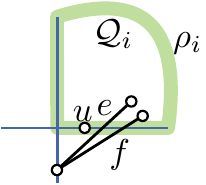}
		}
		\quad
		\subfloat[] {
			\includegraphics[page=2]{figures/2_aligned_structure.pdf}
		}
		
		\caption{The curve $\rho_i$ (a) and its modification in (b).}
		\label{aligned:fig:2aligned_mod}
	\end{figure}

	Let $\rho_i$ be the parts of $\CurX_i$ and $\CurX_{i+1}$ that are on the boundary of
	the cell $\CurQ_i$, see Figure~\ref{aligned:fig:2aligned_mod}. We modify $\rho_i$
	as follows. We first, join the endpoints of $\rho_i$ in the infinity such that
	it becomes a simple closed curve. Let $u$ be a vertex that lies on $\rho_i$.
	We reroute $\rho_i$ such that $u$ now lies outside of $\rho_i$. Since $G_R$ is
	triangulated and $\rho_i$ only intersects edges, $\rho_i$ corresponds to a
	cycle in $G_R^\star$ and therefore to a cut $C_i$ in $G_R$. Note, each edge of
	a connected component in $G-C_i$ is a free edge.

	Now assume that there are two distinct edges $e,f$ that consecutively cross a
	pseudoline $\CurX_i \in \CurX$.  By the premises of the lemma there is a vertex
	that lies on the origin $\CurO$. Hence both $e$ and $f$ cross $\CurX_i$ on the
	same side with respect to $\CurO$.  Since $e$ and $f$ are distinct and $(G_R,
	\CurX)$ is ccw-aligned, there is a cell $\CurQ_j$ such that $\CurQ_j$
	contains two distinct vertices $u$ and $w$ incident to $e$ and $f$,
	respectively. Since $G$ is triangulated and $e$ and $f$ are consecutive along
	$\CurX_i$, $u$ and $w$ are vertices in the same connected component of $G-C_j$.
	Therefore, $(G_R, \CurX)$ contains a free edge.  A contradiction.

	Consider a comb $\mathcal C$ in a cell $\CurQ_i$ that contains two
	distinct vertices $u$ and $v$  in its interior. Since $G$ is triangulated and
	$\mathcal C$ is inclusion-minimal (it does not contain another edge-region),
	$u$ and $v$ belong to the same connected component of $G_R-C_i$. Therefore
	$(G_R,
	\CurX)$ contains a free edge.
\end{proof}

We call a comb \emph{closed} if its two separating edges have the same source vertex.

\begin{lemma}
	%q
  \label{lem:closedEmpty}
  For every reduced aligned triangulation $(G_R,\CurX)$ there is a reduced
  aligned triangulation $(G_R'',\CurX)$ where no closed comb contains a vertex
  such that $(G_R,\CurX)$ has an aligned drawing if $(G_R'',\CurX)$ has an aligned drawing.
\end{lemma}

\begin{proof}
	By Lemma~\ref{lemma:base_case:structure} we know that each comb contains at
	most one vertex.  We apply induction over the number of closed combs that
	contain a vertex.  Let $v$ be a free vertex in a closed comb with separating
	edges $uw_1$, $uw_2$.  Then we obtain an aligned graph $(G_R',\CurX)$ by
	contracting edge $uv$ in the embedding.  Since $(G_R,\CurX)$ is reduced
	ccw-aligned, all edges outgoing from the free vertex $v$ are $1$-anchored
	$0$-crossed or $0$-anchored $1$-crossed.  In $(G_R',\CurX)$ they are now
	2-anchored 0-crossed or 1-anchored 1-crossed with free target vertex. Since
	there is no other vertex in the comb and the comb is closed, $v$ only has $uv$
	as incoming edge which is contracted. Therefore $(G_R',\CurX)$ is
	ccw-aligned.
	Assume that $(G_R', \CurX)$ has an aligned drawing. Since $v$ is a free
	vertex, we obtain an aligned drawing of $(G, \CurX)$ by placing $v$ close to
	$u$ within in its closed comb.
	By Lemma~\ref{lemma:reduction} we obtain a reduced aligned triangulation
	$(G_R'',\CurX)$ from $(G', \CurX)$ such that $(G'_R,\CurX)$ has an aligned
	drawing if $(G_R'',\CurX)$ has an aligned drawing. In the construction the
	number of closed combs that contain a vertex is not increased. 
	%This concludes the induction.
	%
\end{proof}

%For two distinct points $u$, $v$ in the plain we denote the unique line through
%$u$ and $v$ by $\overline{uv}$.

\begin{figure}[tb]
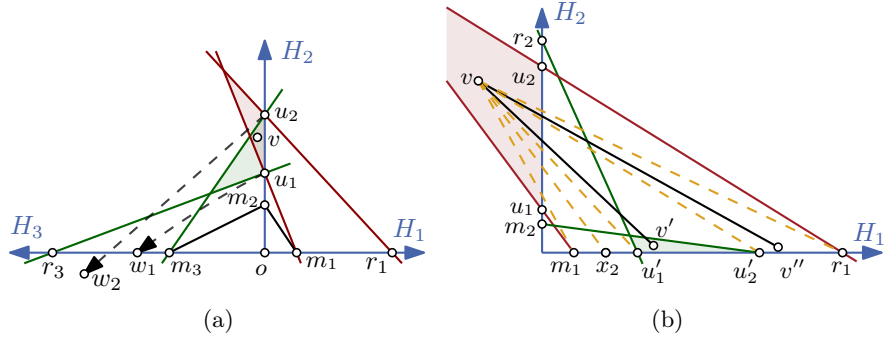

	\centering
	\subfloat[\label{fig:freeDef}]{
		\includegraphics[page=2]{figures/construction_small2.pdf}
	}
	\subfloat[\label{fig:constrObservations}]{
		\includegraphics[page=3]{figures/construction_small2.pdf}
	}
 \caption{(a)~Placement of a free vertex $v$ in cell $\CurQ_2$. It may be
  placed within the gray triangle. (b)~Example for the observations with $u_1'=x_3$ and $u_2'=x_4$.} 
\end{figure}

\begin{figure}
	\centering
	\includegraphics{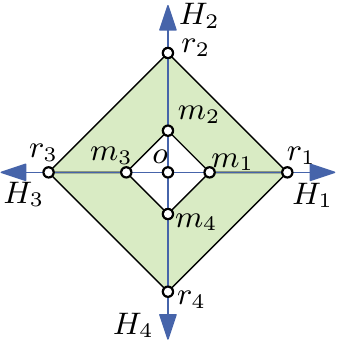}
	\caption{The vertex $o$ and the half-lines $H_i$ and the vertices $m_i$, $r_i$
	for $i=1,\dots,4$. All remaining edges and vertices lie in the green area.}
	\label{fig:HiDef}
\end{figure}

We can now show that each reduced instance has an aligned drawing.

\begin{lemma}
  \label{lemma:base_case:drawing}
	Every reduced ccw-aligned triangulation has an aligned drawing.
\end{lemma}
\begin{proof}
  By Lemma~\ref{lem:closedEmpty} we can assume that in our triangulation $(G,
  \CurX)$ the closed combs contain no vertices. By
  Lemma~\ref{lemma:base_case:structure} we know that each comb contains at most
  one vertex and no vertex if it is closed. The main problem is to draw the
  $1$-crossed edges. For those, we place each free vertex $v$ close to the
  right boundary of its comb. This allows to draw the incoming edges. Since
  $(G,\CurX)$ is ccw-aligned, the target of each $1$-crossed edge $vu$ is free
  and allows to draw $vu$.

	We construct the aligned drawing $(\Gamma, X)$ as follows.
% place corners
	Let $o$ be the vertex on the origin. We call the sources of separating edges
	\emph{corners}. First place $o$ and all corners on $X$ in the order induced
	from $\CurX$.
% place free vertices
	%For the following definitions see Figure~\ref{fig:HiDef}.
	For $i=1,\dots,2|X|$, let $\CurH_i$ be the half-pseudoline that is the right
	boundary of cell $\CurQ_i$. Let $m_i$ denote the vertex on $\CurH_i$ that
	is adjacent to $o$ and let $r_i$ denote the vertex incident to the outer face
	on $\CurH_i$. Note that $m_i$, $r_i$ are corners. We write $u <_i v$ if $u$
	lies between $o$ and $v$ on $\CurH_i$ where $u$, $v$ may be vertices and
	intersections of edges with $\CurH_i$. Note that $<_i$ is a linear order.
	Define $H_i$ correspondingly for $X$; see Figure~\ref{fig:HiDef}. The indices
	for $m_i$, $Q_i$, etc.\ are considered $\bmod\, 2|X|$.  In the following, we
	denote by $\overline{uv}$ the line through two distinct points $u$, $v$.
	%For convenience, set $m_{0}=m_{4}$ and $m_5=m_1$ and similar for cell
	%and related variable names. 	
%
	Now consider a free vertex $v$ in some cell $Q_i$; see
	Figure~\ref{fig:freeDef}. It lies in a comb that is bounded by two separating
	edges $u_1w_1$, $u_2w_2$ with $u_1<_i u_2$ on $\CurH_i$.  Note that we have
	$u_1\ne u_2$ since the comb contains $v$ and is thus not closed.  We place $v$
	within the triangle bounded by $\overline{m_{i+1}u_2}$,
  $\overline{r_{i+1}u_1}$, $H_i$ and between $\overline{m_{i-1}u_1}$,
  $\overline{r_{i-1}u_2}$ (if these lines cross within $Q_i$, then this means within
  the triangle bounded by $\overline{m_{i-1}u_1}$,
  $\overline{r_{i-1}u_2}$, $H_i$). Note that $v$ lies in $Q_i$ .
%place remaining vertices
	We will show that the intersections of $1$-crossed edges with $H_i$ and the
	corners on $H_i$ respect the order $<_i$.  Finally, we place for $i=1,\dots,2|X|$
	the vertices on $\CurH_i$ that are neither $o$ nor a corner arbitrarily on
	$H_i$ respecting the order $<_i$. This finishes the construction (edges are
	placed accordingly).

  We next show that the vertices and edges of $G$ appear for $1\le i\le |X|$
  along $X_i$ and $\CurX_i$ in the same order.
  %We next show that $X$, $Y$ collect the vertices and edges of $G$ in the same
  %order as $\CurX$, $\CurY$.  
  Consider the free vertex $v$ and the separating edges $u_1w_1$, $u_2w_2$ as
  defined above.  Let $m_{i-1}=x_1<_{i-1}\dots<_{i-1}x_k=r_{i-1}$ denote the
  corners on $H_{i-1}$.  The following three observations imply that all
  $1$-crossed edges with target $v$ cross $H_i$ in the correct order between
  $u_1$ and $u_2$; refer to  Figure~\ref{fig:constrObservations}. 

  \begin{enumerate}
    \item \label{itm:constr_obs1} $\overline{m_{i-1}v}$ and
    $\overline{r_{i-1}v}$ cross $H_i$ between $u_1$ and $u_2$.  

  \item \label{itm:constr_obs2} $\overline{x_1v},\dots,\overline{x_kv}$
    intersect $H_i$ in the same order as $x_1,\dots,x_k$ lie on $\CurH_{i-1}$.

  \item \label{itm:constr_obs3}  Let  $v'$ be a free vertex in $Q_{i-1}$.  Let
    $u_1'w_1'$, $u_2'w_2'$ be the separating edges of the comb containing $v'$.
    Then $v'v$ crosses $H_i$ between  $\overline{u_1'v}\cap H_i$ and
    $\overline{u_2'v}\cap H_i$.
  \end{enumerate}
 
	For Observation~\ref{itm:constr_obs1}, note that $v$ lies between
  $\overline{m_{i-1}u_1}$, $\overline{r_{i-1}u_2}$.  For
  Observation~\ref{itm:constr_obs2}, note that
  $\overline{x_1v},\dots,\overline{x_kv}$ cross pairwise in $v$ and thus not in
  section $Q_{i-1}$. These two observations imply that
  $\overline{x_1v},\dots,\overline{x_kv}$ cross $H_{i-1}$ between $u_1$ and
  $u_2$. For Observation~\ref{itm:constr_obs3} note now that $v'$ lies in the
  triangle bounded by $H_{i-1}$, $\overline{u_2'm_i}$ and $\overline{u_1r_i}'$.
  Observation~\ref{itm:constr_obs3} follows from $v$ and this triangle lying between
  $\overline{u_1m_{i-1}}$ and $\overline{u_2r_{i-1}}$.

	We now show that all $1$-crossed edges with target $v$ cross $H_i$ in the
	correct order between $u_1$ and $u_2$.  By
	Observations~\ref{itm:constr_obs2},~\ref{itm:constr_obs3} the $1$-crossed edges
	with target $v$ cross $H_i$ between $\overline{m_{i-1}v}\cap H_i$ and
	$\overline{r_{i-1}v}\cap H_i$. With Observation~\ref{itm:constr_obs1}, they
	cross $H_i$ between $u_1$ and $u_2$.  By Observation~\ref{itm:constr_obs2}, we
	know that the $1$-anchored $1$-crossed edges with target $v$ cross $H_i$ in the
	correct order.  By Observations~\ref{itm:constr_obs2},~\ref{itm:constr_obs3},
	we obtain that each pair of a 0-anchored 1-crossed and a 1-anchored 1-crossed
	edge cross $H_i$ in the correct order.  Since the sources of $0$-anchored
	$1$-crossed edges with target $v$ lie in different combs, they lie pairwise on
	different sides of some edge $x_jv$ by Observation~\ref{itm:constr_obs3}.
	Observation~\ref{itm:constr_obs2} then yields their correct ordering.

  Since the corners on $H_i$ respect $<_i$ and all
  1-crossed edges have free target vertices (as the triangulation is
  ccw-aligned), this implies that the intersections of $1$-crossed edges with
  $H_i$ and the corners on $H_i$ respect the order $<_i$.
  By construction, we placed the vertices on $\CurH_i$ that are not corners
  such that they also respect order $<_i$. Thus the lines $X_j$
  intersect the vertices and edges in the same order as $\CurX_j$.

	We next show that our embedding is planar by showing that there is no location
	where edges cross. Since the order of intersections with lines in $X$ is correct,
	there are no crossings on $X$. This leaves us with the cells.  Since
	the separating edges of $\CurQ_i$ appear in the same order on $\CurH_i$ and
	$\CurH_{i+1}$, they also appear in the same order on $H_i$ and $H_{i+1}$.
	Thus, separating edges of the same cell do not cross each other. We
	further obtain the same combs for $(\Gamma,XY)$.
	Consider again a free vertex $v$ in $Q_i$ and the corresponding separating
	edges $u_1w_1$, $u_2w_2$; see Figure~\ref{fig:freeDef}. Since $v$ lies in the
	triangle bounded by $H_i$, $T_1$ and $\overline{m_{i+1}u_2}$, it also lies in
	the comb bounded by $u_1w_1$, $u_2w_2$.  Hence, every free vertex lies in the
	correct comb. 
	% crossing free
	Let $e$ be an edge incident to $v$. Then its other end vertex does not lie
	within the comb of $v$. It must therefore intersect $\CurH_i$ between $u_1$
	and $u_2$ if it is incoming, and it must intersect $\CurH_{i+1}$ between
	$u_1w_1\cap\CurH_{i+1}$ and $u_2w_2\cap\CurH_{i+1}$ if it is outgoing.  Since
	we have the same order on $H_i$ and $H_{i+1}$ respectively, edge $e$ crosses
	neither $u_1w_1$ nor $u_2w_2$ and thus not the interior of any other comb in
	$Q_i$.  This means that  \begin{inparaenum} \item There are no crossings on
		separating edges in the corresponding cells. And that \item Only edges
			incident to the free vertex $v$ in a comb intersect the interior of that
			comb. Those edges are all adjacent in $v$ and do not cross.
		\end{inparaenum} We obtain that there are no crossings on $X$, no
		crossings on separating edges in the corresponding cells and no
		crossings within combs. Hence, our embedding is planar. 

  Since there are no free edges and the order of intersections with lines in $X$ is
  fixed, the order of incident edges around a free vertex is also fixed.  For a
  vertex $u$ on $X$ we note that each adjacent free vertex is in another
  comb and therefore the order of incident edges around $u$ is also fixed.
  Therefore, our embedding $\Gamma$ induces the same combinatorial embedding as
  the embedding of $G$. 
  \end{proof}

From Lemma~\ref{lemma:reduced} and Lemma~\ref{lemma:base_case:drawing} we
directly obtain our main theorem.

\begin{theorem}\label{aligned:theorem:2aligned:drawable}
	Every ccw-aligned graph $(G, \CurX)$ has an aligned drawing.
\end{theorem}

\section{Parallel lines}

In this section, we prove that every aligned graph $(G, \CurArr)$ has an aligned
drawing, if $\CurArr$ is intersection free, i.e., the line
arrangement $\Arr$ is a set of parallel lines.

Our result uses a result of Eades at
al.~\cite{DBLP:journals/algorithmica/EadesFLN06}, and of Pach and
Toth~\cite{DBLP:journals/jgt/PachT04}. Eades et al.  consider hierarchical plane
graphs. A graph $G=(V, E)$ with a mapping of the vertices to a layer $L_i$ is a
\emph{hierarchical graph}, where a set of \emph{layers} $\mathcal L$ is a set of
ordered parallel horizontal lines $L_i \in \mathcal L$. A hierarchical plane
drawing of a hierarchical graph is a planar drawing where each vertex is on its
desired layer and each edge is drawn as a $y$-monotone curve.  Two hierarchical
drawings are \emph{equivalent} if each layer, directed from $-\infty$ to
$\infty$, crosses the same set of edges and vertices in the same order. Eades
et.  al.~\cite{DBLP:journals/algorithmica/EadesFLN06} proved that for every
hierarchical planar drawing of a graph there is an equivalent hierarchical
planar straight-line drawing.  Pach and Toth~\cite{DBLP:journals/jgt/PachT04}
proved a similar result stating that for every $y$-monotone drawing where no two
vertices have the same $y$-coordinate there is an equivalent $y$-monotone
straight-line drawing such that each vertex keeps its $y$-coordinate. In
contrast to these two results, we have that the $y$-coordinate is only
prescribed for a subset of the vertices, i.e., there are some (free) vertices
that have to be positioned between two layers (lines). The proof strategy is to
extend the initial pseudoline arrangement with an additional set of
intersection-free pseudolines such that there are no free vertices.

Due to \cite{DBLP:journals/jgaa/MchedlidzeRR18} (compare Lemma~\ref{lemma:reduced}), we can assume that there are neither free nor aligned edges. For the
purpose of this section, a \emph{reduced aligned graph} is an aligned graph that
has no aligned edges and no free vertices. Note that previously only free edges
were forbidden. Thus, the current definition is more restrictive.  The following
theorem is an immediate corollary from the results of Eades et
al.~\cite{DBLP:journals/algorithmica/EadesFLN06}, and Pach and
Toth~\cite{DBLP:journals/jgt/PachT04}.

\begin{theorem}
	\label{theorem:proper_drawing}
	For every intersection-free pseudoline arrangement, every reduced aligned graph
	$(G,\CurArr)$ has an aligned drawing. 
\end{theorem}

\begin{lemma}
	\label{lemma:to_proper_aligned}
	Let $\CurArr$ be an intersection-free pseudoline arrangement and let $A$ be
	a line arrangement homeomorphic to $\CurArr$.
	For every aligned graph $(G, \CurArr)$ there is a reduced aligned graph $(G,
	\CurArr')$  such that $\CurArr \subset \CurArr'$ and $(G, \CurArr)$ has an aligned drawing
	if $(G, \CurArr')$ has an aligned drawing.
\end{lemma}

\begin{proof}

\begin{figure}
	\centering
	\includegraphics{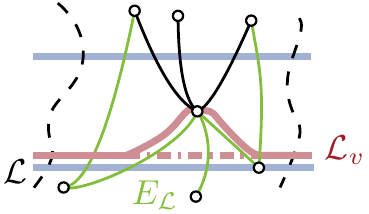}
	\caption{Construction of the new pseudoline $L_v$ (red) that contains $v$.
	The red-dotted pseudoline $L_v'$ indicates the copy of $L$ (bottom blue)
	that crossed the edges in $E_L$ (green) in the same order as $L$}
	\label{fig:pseudoline:ext}
\end{figure}

We first insert for each free vertex $v$ a new pseudoline $\CurL_v$ to
$\CurArr$ such that $v$ is on $\CurL$. Thus, the aligned graph $(G, \CurArr')$
does not have free vertices.

Let $\CurL$ be a pseudoline that is on the boundary the region $R_v$ of
$\CurArr$ that contains $v$.  Let $E_\CurL$ be the set of edges of $G$ that are
(partially) routed through $R_v$ and that are either crossed by $\CurL$ or that
have an endpoint on $\CurL$. 
We assume that $\CurL$ is directed.  Then the direction of $\CurL$ induces a
total order of the edges in $E_\CurL$. We obtain a curve $\CurL_v'$ that crosses
the edges in $E_\CurL$ in this order and in their interior. Since $v$ is free, $G$ is
triangulated and $(G,\CurArr)$ contains neither free nor aligned edges, there is
at-least one edge $e \in E_\CurL$ that is incident to $v$. Denote by $e_f$ and $e_l$ in
$E_\CurL$ the first and last edge incident to $v$. We obtain a pseudoline
$\CurL_v$ that contains $v$ from $\CurL_v'$ by rerouting $\CurL_v'$ along $e_f$
and $e_l$ such that it is does not cross these edges in their interior and such
that $v$ is on the line (Fig.~\ref{fig:pseudoline:ext}).

Now, let $(G, \CurArr')$ be the aligned graph that is obtained by the previous
procedure for each free vertex $v$.
Let $A'$ be any set of parallel lines that contains $A$ and corresponds to
$\CurArr'$.
Clearly, $(\Gamma, A)$ is an aligned drawing of $(G, \CurArr)$ if $(\Gamma, A')$
is an aligned drawing of $(G, \CurArr')$. This finishes the proof.
\end{proof}

Theorem~\ref{theorem:proper_drawing} and Lemma~\ref{lemma:to_proper_aligned} together prove
the following theorem.

\begin{theorem}
	Let $\CurArr$ be an intersection-free pseudoline arrangement and let $A$ be
	a (parallel) line arrangement homeomorphic to $\CurArr$. Then every aligned graph $(G,
	\CurArr)$ has an aligned drawing $(G, A)$.
\end{theorem}

\section{Conclusion}

In the paper, we showed that every aligned graph $(G, \CurArr)$ has an aligned
drawing if $(G, \CurArr)$ is either a ccw-aligned graph or if $\CurArr$ is
intersection-free. Further, we provided a non-trivial example of a $2$-star
aligned graph that does not admit an aligned drawing. Thus, in our opinion the
most intriguing open question is whether every aligned graph of alignment
complexity $(1,0,0)$ has an aligned drawing, for general stretchable pseudoline
arrangements $\CurArr$. Our example shows that this statement is not
true for aligned graphs of alignment complexity $(1,1,0)$. Our stretchability
proof of counterclockwise aligned graphs uses the fact that we can move each
free vertex $v$ to an aligned vertex $u$ on the cell of $v$. Performing this
operation for all free vertices at once ensures that we do not introduce edges
of a forbidden alignment complexity. Figure~\ref{fig:vertex_mapping} indicates
that for general aligned graphs of alignment complexity $(1,0,0)$ there is not
always a consistent mapping of free vertices to aligned vertices such that that
the resulting graph has the same alignment complexity. Thus it is unclear
whether the techniques used in the paper can be used to decide whether every
aligned graph of alignment complexity $(1,0,0)$ has an aligned drawing.

\begin{figure}
	\centering
	\subfloat[]{
		\includegraphics[page=1]{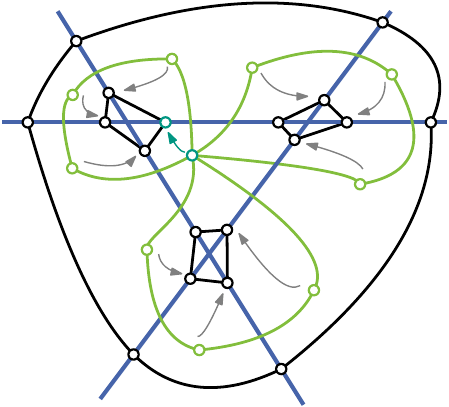}
	}
	\subfloat[]{
		\includegraphics[page=2]{figures/general_k.pdf}
	}
	\caption{There is no mapping of free vertices to aligned vertices on the
	boundary of the same cell such that moving the free vertices onto their image
	results in an aligned graph of alignment complexity $(1,0,0)$.}
	\label{fig:vertex_mapping}
\end{figure}

%%
%% Bibliography
%%
\bibliography{strings,references}

\appendix
\section{Reducing $k$-star aligned graphs}

In this section, we give further details on how to reduce a $k$-star aligned
graph. We first recall the definition of \emph{proper} and \emph{reduced}
triangulations.
A $k$-star aligned graph $(G, \CurX)$ is a \emph{proper $k$-aligned
triangulation} if each inner face is a triangle, the boundary of the outer face
is a $2k$-cycle of $2$-anchored edges, the outer face does not contain the
origin and there is a degree-$2k$ vertex $o$ on the origin incident to $2k$
aligned edges.  We refer to a reduced proper ccw-aligned triangulation as a
\emph{reduced aligned triangulation} if it does neither contain
\begin{inparaenum}[(i)]
	\item separating triangles, nor 
	\item free edges, nor
	\item aligned edges that are not incident to the origin $\mathcal O$.)
\end{inparaenum}

\begin{figure}
	\centering
	\subfloat[]{\includegraphics[page=3]{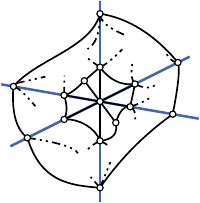}}
	\quad
	\subfloat[]{\includegraphics[page=4]{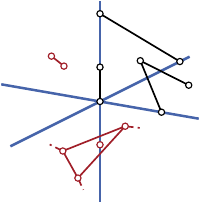}}
	\quad
	\caption{(a)~Illustration of the key properties of a proper $k$-star aligned
	graph.(b)~Examples of allowed (black) edges in a reduced instance and
forbidden (red) edges.}
\end{figure}

\begin{figure}
		\centering
		%\subfloat[\label{fig:tri:outer_face}]{
		%  \includegraphics[page=1]{figures/triangulation.pdf}
		%}
		%\quad
		\subfloat[\label{fig:tri:sep_edge}]{
			\includegraphics[page=5]{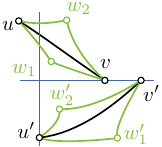}
		
		}
		\quad
		\subfloat[]{
			\includegraphics[page=6]{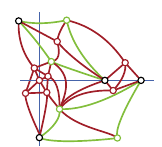}
		
		}
		\quad
		\subfloat[\label{fig:tri:part}]{
			\includegraphics[page=7]{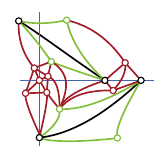}
		}

		\caption{(a)~The (black) separating edges are isolated by the green edges. (b) The
		black edges are removed and the red edges are obtained by the triangulation.
		(c) Final graph, after removing edges in the interior of a quadrangle $u,
		w_1, v, w_2$ and reinserting the black edges.}
		\label{fig:triangulation}
\end{figure}

Mchedldize et al. proved the following triangulation lemma.

\begin{lemma}
	\label{lemma:old_tri}
	For every aligned graph $(G, \CurX)$ of alignment complexity $(1,0,\sentinel)$
	there is an aligned triangulation $(G', \CurX)$ of alignment complexity
	$(1,0,\sentinel)$ such that $G$ is a subgraph of $G'$.
\end{lemma}

Since ccw-aligned graphs contain $2$-anchored and $1$-anchored $1$-crossed
edges, we can not immediately apply this lemma. In the following, we show that
our instances can be modified such that the can use the previous lemma. For
simplicity, we assume that there is no aligned that crosses the origin. 

\begin{lemma} \label{lemma:tri}
	Let $(G, \CurX)$ be a ccw-aligned graph. Then there is a ccw-aligned
	triangulation $(G', \CurX)$ that contains $(G, \CurX)$ as a subgraph.
	Moreover, the outer face of $(G', \CurX)$ is bounded by $2k$-cycle $C$ of
	$2$-anchored edges and the outer face does not contain the origin in its
	interior.
\end{lemma}

\begin{proof}
 
	Let $(G_2, \CurX)$ be the graph that is constructed from $(G, \CurX)$ as
	follows. First, add a $2k$-cycle~$C$ of $2$-anchored edges in the outer face
	such that the new outer face does not contain the origin. 

	For each separating edge $uv$ of $G$ add two vertices $w_1, w_2$ and the edges
	$uw_1, w_1v$ and $uw_2, w_2v$. Route and direct the edges according to
	Figure~\ref{fig:tri:sep_edge}.  Finally, remove the edge $uv$. Eventually, we
	arrive at an aligned graph of alignment complexity  $(1,0, \sentinel)$.
	With the application of Lemma~\ref{lemma:old_tri} we obtain a
	triangulated aligned graph $(G_3, \CurX)$ of alignment complexity $(1, 0,
	\sentinel)$.  We remove edges in the interior of each quadrangle $u, w_1, v,
	w_2$ and reinserted the original edge $uv$. Finally, we remove all edges and
	vertices in the region bounded by $C$ that does not contain the origin. 	
	This
	yields the desired aligned graph $(G', \CurX)$.
\end{proof}

Since no free edge of an ccw-aligned graph is incident to a triangle that
contains the intersection in its interior, the next lemma follows from the
results of Mchedlitze et al.

\begin{lemma}\label{lemma:contraction}
	Let $(G, \CurX)$ be a ccw-aligned graph and let $e$ be an interior free edge
	or an aligned edge that is neither an edge of a separating nor a chord and does not
	contains the origin, then $(G/e, \CurX)$ is a ccw-aligned graph and $(G,
	\CurX)$ has an aligned drawing if $(G/e, \CurX)$ has an aligned drawing.
\end{lemma}

Thus, we can now prove the main reduction lemma and therefore
Lemma~\ref{lemma:reduced}.

\begin{lemma}
	\label{lemma:reduction}
	For every ccw-aligned graph $(G, \CurX)$ there is a reduced aligned triangulation
	$(G_R,
	\CurX)$  such that $(G, \CurX)$ has an aligned drawing if $(G_R, \CurX)$ has
	an aligned drawing.
\end{lemma}

\begin{proof}
	\begin{figure}
		\centering
		\subfloat[\label{fig:origin:edge}]{
			\includegraphics[page=1]{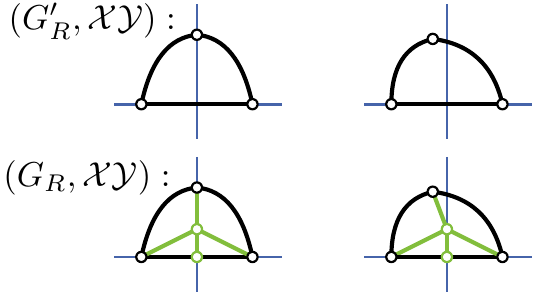}
		}
		\quad
		\subfloat[\label{fig:origin:vertex}]{
			\includegraphics[page=2]{figures/intersection.pdf}
		}
		\caption{Red edges are removed from $(G_T, \CurX)$ and green added
		to $(G_P, \CurX)$}
		\label{fig:origin}
	\end{figure}

	By Lemma~\ref{lemma:tri} there is a aligned triangulation $(G_T, \CurX)$ of
	$(G, \CurX)$ with the outer face bounded by $2k$-cycle of $2$-anchored edges.
	Moreover, an aligned drawing of $(G_T, \CurX)$ contains an aligned drawing of
	$(G, \CurX)$.
	
	By Mchedlidze et al. we obtain a reduced aligned
	triangulation $(G_R', \CurX)$ from $(G_T, \CurX)$ by either splitting $(G_T,
	\CurX)$ into two aligned graphs at a separating triangle $T$, or by
	contracting free or aligned edges that are not incident to $o$ (Lemma~\ref{lemma:contraction}). 
	Moreover, we have that that $(G_T, \CurX)$ has an
	aligned drawing if $(G_R', \CurX)$ has an aligned drawing 

	In order to obtain a proper aligned triangulation $(G_R, \CurX)$ from $(G'_R,
	\CurX)$ we perform the reduction depicted in Figure~\ref{fig:origin}.  If
	there is an aligned edge that contain the origin in its interior,  we place a
	subdivision vertex on this edge and inserted edges as depicted in
	Figure~\ref{fig:origin:edge}.  Note that in this case an aligned drawing of
	$(G_R, \CurX)$ contains an aligned drawing of $(G_R', \CurX)$.

	Consider the case that there is a vertex $v$ on the origin that is incident to
	a free vertex $u$. We obtain a new aligned graph $(G_R, \CurX)$ by
	exhaustively applying the reductions depicted in
	Figure~\ref{fig:origin:vertex}.
	Since the black polygon (compare Figure~\ref{fig:origin:vertex}) in an aligned
	drawing of $(G_R, \CurX)$ is star-shaped  and its kernel contains the
	vertex~$v$, $(G_R', \CurX)$ has an aligned drawing if $(G_R', \CurX)$ has an
	aligned drawing. 
\end{proof}

\end{document}